\documentclass{article}
\usepackage{graphicx}
\usepackage{spconf}
\usepackage{color}%
\usepackage{float}
\usepackage{amsmath,bm,amssymb,amsfonts,amsthm,epsfig}
\usepackage{psfrag}
\usepackage{comment}
\usepackage{url}
\usepackage{cite} 
\usepackage{lipsum}
\newcounter{opteq}
\newenvironment{opteq}{\refstepcounter{opteq}\align}{\tag{P\theopteq}\endalign}

\makeatletter
\newcommand*{\rom}[1]{\expandafter\@slowromancap\romannumeral #1@}
\makeatother

\newtheorem{thm}{Theorem}

\usepackage{algorithm}
\usepackage{algorithmicx}
\usepackage[noend]{algpseudocode}

\usepackage{arydshln}
\usepackage{mathrsfs}
\usepackage{bbm}

\usepackage[toc,page]{appendix}

\newcommand{\txt}[1]{\text{\normalfont #1}}

\makeatletter
\newcounter{savesection}
\newcounter{apdxsection}
\renewcommand\appendix{\par
	\setcounter{savesection}{\value{section}}%
	\setcounter{section}{\value{apdxsection}}%
	\setcounter{subsection}{0}%
	\gdef\thesection{\@Alph\c@section}}
\newcommand\unappendix{\par
	\setcounter{apdxsection}{\value{section}}%
	\setcounter{section}{\value{savesection}}%
	\setcounter{subsection}{0}%
	\gdef\thesection{\@arabic\c@section}}
\makeatother

\title{Analog beamforming for active imaging using sparse arrays}
\name{\begin{tabular}{c} Robin Rajam\"{a}ki$^{\dagger}$,
		~Sundeep Prabhakar Chepuri $^\star$,
		and Visa Koivunen$^\dagger
		$\end{tabular}}
\address{$^\dagger$ Aalto University, Espoo, Finland \\
	$^{\star}$ Indian Institute of Science, Bangalore, India
}

\begin{document}
	\setlength\belowcaptionskip{-15ex}
	%
	\maketitle
	\begin{abstract}		
		This paper studies analog beamforming in active sensing applications, such as millimeter-wave radar or ultrasound imaging. Analog beamforming architectures employ a single RF-IF chain connected to all array elements via inexpensive phase shifters. This can drastically lower costs compared to fully-digital beamformers having a dedicated RF-IF chain for each sensor. However, controlling only the element phases may lead to elevated side-lobe levels and degraded image quality. We address this issue by image addition, which synthesizes a high resolution image by adding together several lower resolution component images. Image addition also facilitates the use of sparse arrays, which can further reduce array costs. To limit the image acquisition time, we formulate an optimization problem for minimizing the number of component images, subject to achieving a desired point spread function. We propose a gradient descent algorithm for finding a locally optimal solution to this problem. We also derive an upper bound on the number of component images needed for achieving the traditional fully-digital beamformer solution.
	\end{abstract}
	%
	\section{Introduction}
    The use of high frequencies enables small array form factors by packing many elements into a small physical area. For example, 3D ultrasound imaging typically uses hundreds of sensors, each with a dedicated transceiver chain. Although the resulting large electrical aperture improves the array's resolution, the hardware cost, number of cables, power consumption, and computational load of the array may become prohibitively large. These issues are especially pronounced for fully-digital arrays, where each array element is connected to a separate RF-IF front-end and \emph{analog-to-digital converter} (ADC) or \emph{digital-to-analog converter} (DAC). 
    
    \emph{Sparse arrays} can significantly reduce the number of elements compared to uniform arrays of equivalent aperture, without sacrificing the array's ability to resolve scatterers \cite{hoctor1990theunifying,pal2010nested,wang2017coarrays}. Such arrays utilize the virtual \emph{co-array} structure consisting of the pairwise vector sums or differences of the physical array element positions. For instance, the \emph{sum co-array} determines the achievable set of \emph{point spread functions} (PSFs) of an active far field imaging array that uses linear processing at the transmitter and receiver \cite{hoctor1990theunifying}. A desired PSF may be synthesized using the \emph{image addition} technique \cite{hoctor1990theunifying,kozick1991linearimaging}, which adds together several component images corresponding to different transmit-receive element weight pairs. The number of component images should be kept as low as possible, to minimize the image acquisition time when transmitters operate coherently, as in a phased array.
    
    \emph{Hybrid beamforming} is another approach for reducing the array costs \cite{molisch2017hybrid}. Hybrid architectures reduce the number of RF-IF chains by pre-processing the antenna signals by a network of inexpensive low power phase shifters connecting every antenna to each front-end. Hybrid beamformer design has been extensively studied for millimeter-wave communications, where linear processing is used to precode and decode multiple data streams sent over a multiple-input multiple-output channel \cite{zhang2005variable,zhang2014onachieving,yu2016alternating,jin2018hybridprecoding,koochakzadeh2018beam,sohrabi2016hybrid,bogale2016onthenumber}, and hybrid beamformers are designed to maximize the channel capacity \cite{alkhateeb2014mimo}. In contrast, this paper considers \emph{active imaging} applications, where synthesizing a desirable PSF is of main interest \cite{hoctor1990theunifying}. Furthermore, we focus on the extreme case of \emph{fully-analog beamforming}, where all the array elements are connected to a single RF-IF-ADC/DAC chain by phase shifters with continuous phases.\footnote{We address the more general case of hybrid beamforming with quantized phase shifts in the longer journal version of this paper \cite{rajamaki2019hybrid}.} We utilize image addition to synthesize PSFs that are commonly achieved only by fully-digital arrays. Image addition also facilitates the use of sparse transmitting and receiving arrays, which further reduce the required number of array elements thereby reducing the hardware costs.
    
    The contributions of the paper are threefold. First, we formulate an optimization problem, where we design the analog transmit and receive beamforming weights achieving a desired PSF using as few component images as possible. Second, we develop a gradient descent algorithm for finding a local minimum of this non-convex problem. Third, we derive an upper bound on the number of component images, and give the beamformer weights achieving this bound in closed-form.
    
	\begin{figure*}[t]
		\centering
		\includegraphics[width=1\textwidth]{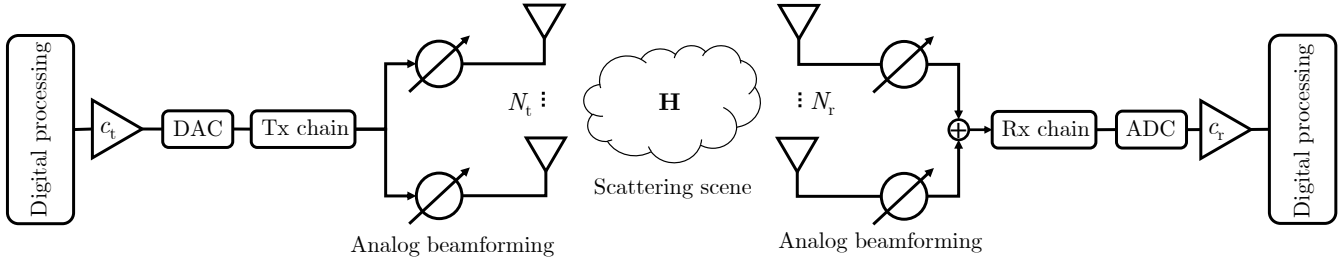} 
		\caption{Fully-analog beamforming architecture. The RF-IF front-end of the transmitter/receiver is connected to each array element via a phase shifter with continuous phase.}\label{fig:architecture}
	\end{figure*}

	\section{Signal model and definitions} \label{sec:definitions}
	Consider a sensor array with $N_\txt{t}$ transmit (Tx) and $N_\txt{r}$ receive (Rx) elements. As shown in Fig.~\ref{fig:architecture}, each array element is connected to a single RF-IF front-end and DAC or ADC via a phase shifter. This is referred to as a fully-analog beamforming architecture, in contrast to the fully-digital architecture, where each sensor has a dedicated front-end and DAC/ADC.

	We consider a phased array imaging system that sequentially scans a scattering scene by transmitting and receiving focused beams of narrowband signals. Such systems are typically used in, e.g., medical ultrasound imaging or radar. On the transmitter side, beamforming is achieved by combining the Tx array element outputs with weights $ \mathbf{w}_\txt{t}=c_t\mathbf{f}_\txt{t}\in\mathbb{C}^{N_\txt{t}} $, where $ c_\txt{t} \in \mathbb{C}$ is a digital transmit gain, and $ \mathbf{f}_\txt{t}\in \mathbb{C}^{N_\txt{t}} $ is a vector of transmitter phase shifts. Similarly, on the receiver side, the Rx element inputs are combined using the weights $ \mathbf{w}_\txt{r}=c_r\mathbf{f}_\txt{r}\in\mathbb{C}^{N_\txt{r}} $, where $ c_\txt{r} \in \mathbb{C} $ is a digital receive gain, and $ \mathbf{f}_\txt{r}\in \mathbb{C}^{N_\txt{r}} $ is a vector of receiver phase shifts. The received signal after beamforming and matched filtering is modeled as
		\begin{align}
	y &=  \mathbf{w}_\txt{r}^\txt{T}\mathbf{H}\mathbf{w}_\txt{t}+\mathbf{w}_\txt{r}^\txt{T}\mathbf{n},\label{eq:gamma_hat}
	\end{align}
	where $ \mathbf{H} \in\mathbb{C}^{N_\txt{r}\times N_\txt{t}}$ is the channel matrix, $ \mathbf{n}\in\mathbb{C}^{N_\txt{r}} $ is the receiver noise vector, and $ (\cdot)^\txt{T} $ denotes transposition.
		
	The \emph{point spread function} (PSF) is an important property characterizing an imaging system. The PSF is defined as the system's response to a unit-reflectivity, line-of-sight point scatterer. Specifically, for a scatterer located at $ \mathbf{v} \in\mathbb{R}^3$ and the array focused at $ \mathbf{u} \in\mathbb{R}^3$, the PSF is defined as $\psi(\mathbf{u},\mathbf{v}) = \mathbf{w}_\txt{r}^\txt{T}(\mathbf{u})\mathbf{a}_\txt{r}(\mathbf{v})\mathbf{a}^\txt{T}_\txt{t}(\mathbf{v})\mathbf{w}_\txt{t}(\mathbf{u}) \in \mathbb{C}$, where $ \mathbf{a}_\txt{x} $ is the steering vector, and subscript $\txt{x}\in\{\txt{t},\txt{r}\} $ is used to denote both the transmitter and receiver. If the scatterer is in the far field of linear Tx and Rx arrays, $ \mathbf{v} $ and $ \mathbf{u} $ simplify to angles in $ [-\pi/2,\pi/2] $. We henceforth omit the explicit dependence of $ \mathbf{u} $ and $ \mathbf{v} $ for notational convenience. The PSF may also be expressed using the Kronecker product $ \otimes $ as:
	\begin{align}
	\psi &=(\mathbf{a}_\txt{t}\otimes \mathbf{a}_\txt{r})^\txt{T}\text{vec}(\mathbf{w}_\txt{r}\mathbf{w}^\txt{T}_\txt{t}).\label{eq:psf}
	\end{align}
	To form an image, we steer the Tx and Rx arrays over a discrete set of directions $\{\mathbf{u}_i\}_{i=1}^{U}$ and measure reflectivities in \eqref{eq:gamma_hat} at each image pixel. However, a single Tx-Rx weight pair $ \mathbf{w}_\txt{t},\mathbf{w}_\txt{r} $, as in \eqref{eq:psf}, may not always suffice to achieve a desired PSF. In this case, the image quality may be improved by image addition \cite{hoctor1990theunifying}. Image addition synthesizes a high-resolution composite image by summing together several component images of lower resolution, which are formed using different Tx-Rx weight pairs. With image addition, the rank-1 matrix $\mathbf{w}_\txt{r}\mathbf{w}_\txt{t}^\txt{T}$ in \eqref{eq:psf} is replaced by the \emph{co-array weight matrix} $\mathbf{W}\in \mathbb{C}^{N_\txt{r}\times N_\txt{t}} $ defined \cite{kozick1991linearimaging}:
	\begin{align}
		\mathbf{W} &= \sum_{q=1}^{Q} \mathbf{w}_{\txt{r},q}\mathbf{w}_{\txt{t},q}^\txt{T}. \label{eq:W_d}
	\end{align}
	Each outer product $ \mathbf{w}_{\txt{r},q}\mathbf{w}_{\txt{t},q}^\txt{T} $ in \eqref{eq:W_d} corresponds to a transmission and reception with a different pair of effective Rx and Tx weight vectors $ \mathbf{w}_{\txt{r},q}$ and $\mathbf{w}_{\txt{t},q}$. These vectors may be retrieved from the SVD of matrix $ \mathbf{W} $ in the case of fully-digital beamforming \cite{kozick1991linearimaging}. The smaller the \emph{number of component images} $ Q $ is, the shorter the image acquisition time, as fewer transmissions/receptions are required to form an image. In the case of analog beamforming, \eqref{eq:W_d} can be decomposed as
	\begin{align}
		\mathbf{W}&= \sum_{q=1}^{Q} c_{\txt{r},q}c_{\txt{t},q}\mathbf{f}_{\txt{r},q}\mathbf{f}_{\txt{t},q}^\txt{T} 
		= \mathbf{F}_\txt{r}\txt{diag}(\mathbf{c})\mathbf{F}_\txt{t}^\txt{T}, \label{eq:W_a}
	\end{align}
	where $ \mathbf{c}=\mathbf{c}_\txt{t}\circ \mathbf{c}_\txt{r}\in\mathbb{C}^{Q}$ is the Hadamard product of the digital transmit and receive weights $\mathbf{c}_\txt{x} = [c_{\txt{x},1},\ldots, c_{\txt{x},Q}]^\txt{T}$. Matrix $ \mathbf{F}_\txt{x} =[\mathbf{f}_{\txt{x},1},\ldots,\mathbf{f}_{\txt{x},Q}] \in \mathscr{F}_\txt{x}$ contains the phase shift vectors of the component images, and
	\begin{align}
	\mathscr{F}_\txt{x} = \big\{\mathbf{F}=\exp(j\boldsymbol{\Phi})\ |\ \boldsymbol{\Phi}\in [0,2\pi)^{N_\txt{x}\times Q}\big\} \label{eq:F_set}
	\end{align}
	denotes the set of such phase shift matrices. If we know $ \mathbf{c} $, we may set $ \mathbf{c}_\txt{t} = \mathbf{1}_Q $ and $ \mathbf{c}_\txt{r}=\mathbf{c} $, where $ \mathbf{1}_Q $ is a $ Q $-dimensional vector of ones. This simple choice also maximizes the transmit power under the constraint $ |c_{\txt{t},q}|\leq 1, q=1,2,\ldots,Q $.
	
	\section{Bounds on no. of component images $ Q $} \label{sec:analytical}
	Next, we derive an upper and lower bound on the number of component images $ Q $ required by an analog beamformer for factorizing any co-array matrix $ \mathbf{W} \in\mathbb{C}^{N_\txt{r}\times N_\txt{t}}$ as in \eqref{eq:W_a}. In the case of fully-digital beamforming, SVD can be used to decompose $ \mathbf{W} $ as in \eqref{eq:W_d} using $Q_\txt{d}= \text{rank}(\mathbf{W})\leq\min(N_\txt{r},N_\txt{t}) $ component images \cite{kozick1991linearimaging}. Any analog factorization in \eqref{eq:W_a} must therefore satisfy $ Q\geq Q_\txt{d} \geq \min(N_\txt{r},N_\txt{t})$, where $ Q_\txt{d} $ is the number of component images of the fully-digital beamformer.
	
	An upper bound on $ Q $ may be derived based on the fact that any $ \mathbf{w}_\txt{x}\in\mathbb{C}^{N_\txt{x}}$ can be factorized as $ \mathbf{w}_\txt{x} =\mathbf{F}_\txt{x}\mathbf{c}_\txt{x} $, where $ \mathbf{c}_\txt{x}\in\mathbb{C}^2 $ and $ \mathbf{F}_\txt{x}\in\mathscr{F}_\txt{x}\subset \mathbb{C}^{N_\txt{x}\times 2}$ following \eqref{eq:F_set} \cite[Theorem~1]{zhang2005variable}. Consequently, given a fully-digital factorization of $ \mathbf{W} $ using $ Q_\txt{d} $ component images, we may construct a fully-analog factorization of the same $ \mathbf{W} $ using $ Q= 4Q_\txt{d} $ component images. In particular, the phase shifts and digital weights are given by the following theorem:
	\begin{thm}[Upper bound on $ Q $]\label{thm:M1F1Jinf}
		Any $ \mathbf{W}=\sum_{\tilde{q}=1}^{Q_\txt{d}}\mathbf{w}_{\txt{r},\tilde{q}}\mathbf{w}_{\txt{t},\tilde{q}}^\txt{T} \in\mathbb{C}^{N_\txt{r}\times N_\txt{t}}$ may be factorized as $ \mathbf{W} = \sum_{{q} = 1}^{4Q_\txt{d}}{c}_{\txt{r},{q}}{c}_{\txt{t},{q}}\mathbf{f}_{\txt{r},{q}}\mathbf{f}_{\txt{t},{q}}^\txt{T} $, with $ {c}_{\txt{x},q}\in\mathbb{C} $; and $ \mathbf{f}_{\txt{x},q}\in \mathscr{F}_\txt{x}\subset \mathbb{C}^{N_\txt{x}}$ following \eqref{eq:F_set}. For example, a valid factorization is 
		\begin{align}
		\mathbf{f}_{\txt{x},{q}} &= \exp(j\boldsymbol{\phi}_{\txt{x},\tilde{q}})\label{eq:bx_analog_inf}\\
		{c}_{\txt{x},{q}}&=\|\mathbf{w}_{\txt{x},\tilde{q}}\|_\infty/2, \label{eq:gx_analog_inf}
		\end{align}
		where $\boldsymbol{\phi}_{\txt{x},\tilde{q}}\!=\!\measuredangle \mathbf{w}_{\txt{x},\tilde{q}}+(-1)^{i_\txt{x}+1}\cos^{-1}(|\mathbf{w}_{\txt{x},\tilde{q}}| \|\mathbf{w}_{\txt{x},\tilde{q}}\|^{-1}_\infty)$; $ \tilde{q}= \lceil{q}/4 \rceil$; $i_\txt{r} = \lceil (1+({q}-1) \bmod 4)/2\rceil $; and $ i_\txt{t}=1+({q}-1) \bmod 2$. Here $ \measuredangle $, $ \cos^{-1}, $ and $ |\cdot| $ are applied elementwise.
	\end{thm}
	\begin{proof}[Proof sketch (see \cite{rajamaki2019hybrid} for details)]
		By \cite[Theorem~1]{zhang2005variable}, any $ \mathbf{w}_\txt{x}=\sum_{m=1}^2 c_{\txt{x},m}\mathbf{f}_{\txt{x},m} $. Consequently, $ \mathbf{w}_\txt{r}\mathbf{w}_\txt{t}^\txt{T} = \sum_{m=1}^4 c_{\txt{r},i}c_{\txt{t},l}\\ \mathbf{f}_{\txt{r},i}\mathbf{f}_{\txt{t},l}^\txt{T}$, where $ i $ and $ l $ are functions of the summation index $ m $. As $ \mathbf{W} $ is a sum of $ Q_\txt{d} $ rank-1 matrices, we have $ Q=4Q_\txt{d} $. 
	\end{proof}
	The solution given by Theorem~\ref{thm:M1F1Jinf} is actually non-unique, and it is possible to factorize \eqref{eq:W_a} using another set of $ Q=4Q_\txt{d} $ component weights. Nevertheless, the choice made in \eqref{eq:bx_analog_inf} and \eqref{eq:gx_analog_inf} is particularly simple, and it enables us to reduce the number of component images to $ Q=Q_\txt{d} $, at the expense of doubling the number of phase shifters connected to each array element. This follows from the observation that any $ \mathbf{w}_\txt{x}\in\mathbb{C}^{N_\txt{x}} $ that can be factorized as $\mathbf{w}_\txt{x}=c_\txt{x}\mathbf{F}_\txt{x}\mathbf{1}_2 $, where $ c_\txt{x}\in\mathbb{C} $, and $ \mathbf{F}_\txt{x}\in\mathscr{F}_\txt{x}\subset \mathbb{C}^{N_\txt{x}\times 2}$ following \eqref{eq:F_set}, can be implemented by an analog architecture using $ 2N_\txt{x} $ phase shifters \cite{zhang2014onachieving,sohrabi2016hybrid,bogale2016onthenumber}. However, this requires a modification to the architecture in Fig.~\ref{fig:architecture}, and will therefore not be considered henceforth.
	
	\section{Problem formulation}\label{sec:problem}
 	The goal of the optimization problem formulated in this paper is to minimize the number of component images $ Q$, while achieving a desired PSF. Assuming that the PSF is evaluated for a set of $ V $ discrete target directions $ \{\mathbf{v}_i\}_{i=1}^V $, we may express the desired PSF as $ \boldsymbol{\psi} \in\mathbb{C}^{V}$ and the realized PSF as $ \mathbf{A}\txt{vec}(\mathbf{W}) $. The $ i $th row of measurement matrix $ \mathbf{A}\in\mathbb{C}^{V \times N_\txt{r}N_\txt{t}} $ corresponds to vector $ \mathbf{a}_\txt{t}^\txt{T}(\mathbf{v}_i)\otimes \mathbf{a}_\txt{r}^\txt{T}(\mathbf{v}_i) $. Since the analog beamforming architecture imposes the constraint that $ \mathbf{W}$ should be factorized as \eqref{eq:W_a}, the vectorized $ \mathbf{W} $ may also be expressed using the Khatri-Rao product $ \diamond $ as $ \txt{vec}(\mathbf{W}) = (\mathbf{F}_\txt{t} \diamond \mathbf{F}_\txt{r})\mathbf{c} $. This leads us to formulate the \emph{analog beamformer weight optimization problem}:
	\begin{opteq}
		\underset{\mathbf{F}_{\txt{x}}\in \mathscr{F}_\txt{x},\mathbf{c}\in\mathbb{C}^{Q},Q\in\mathbb{N}_+}{\text{minimize}}\ &Q\nonumber \\
		\text{subject to}\qquad  &\|\boldsymbol{\psi} - \mathbf{A}(\mathbf{F}_\txt{t}\diamond \mathbf{F}_\txt{r})\mathbf{c}\|_2\leq \varepsilon_{\max}, \label{p:h}
	\end{opteq}
	where $ \varepsilon_{\max} \in\mathbb{R}_+ $ is an error tolerance. The fact that $ Q $ is unknown adds an extra layer of complexity to problem \eqref{p:h}. Fixing $ Q $ yields the following simpler optimization problem:
	\begin{opteq}
		\underset{\mathbf{F}_\txt{x}\in \mathscr{F}_\txt{x},\mathbf{c}\in\mathbb{C}^{Q}}{\text{minimize}}\ & \|\boldsymbol{\psi}-\mathbf{A}(\mathbf{F}_\txt{t}\diamond \mathbf{F}_\txt{r})\mathbf{c}\|_2^2.\label{p:h_alt}
	\end{opteq}
	If we can solve \eqref{p:h_alt}, we can easily recover the solution to \eqref{p:h} by finding the smallest $ Q $ for which the objective of \eqref{p:h_alt} does not exceed $ \varepsilon_{\max}^2 $. Note that in practice, the maximum value of $ Q $ is determined by Theorem~\ref{thm:M1F1Jinf}, or by a design constraint on the minimum imaging frame rate.
		
	\section{Gradient Descent algorithm}\label{sec:algorithms}
	In this section, we present a simple gradient descent method for solving \eqref{p:h_alt}. We start by noting that the optimal value of $ \mathbf{c} $ in \eqref{p:h_alt} is the least-squares solution $ \mathbf{c} =(\mathbf{A}(\mathbf{{F}}_\txt{t}\diamond \mathbf{{F}}_\txt{r}))^\dagger\boldsymbol{\psi} $, where $ \dagger $ denotes the pseudo-inverse. We also write the analog weight matrix $ \mathbf{{F}}_\txt{x} $ directly as a function of the unknown phase matrix $ \boldsymbol{\Phi}_\txt{x}\in\mathbb{R}^{N_\txt{x}\times Q}  $, i.e., $ \mathbf{{F}}_\txt{x}(\boldsymbol{\Phi}_\txt{x}) = \exp(j\boldsymbol{\Phi}_\txt{x}) $, where we apply the exponential function elementwise. Similar to \cite{jin2018hybridprecoding}, we then express \eqref{p:h_alt} in terms of variables $ \boldsymbol{\Phi}_\txt{r},\boldsymbol{\Phi}_\txt{t} $ as:
	\begin{opteq}
		\underset{\boldsymbol{\Phi}_\txt{x}\in\mathbb{R}^{N_\txt{x}\times Q}}{\text{minimize}}	\underbrace{\|(\mathbf{I}_V-\mathbf{K}(\boldsymbol{\Phi}_\txt{t},\boldsymbol{\Phi}_\txt{r})\mathbf{K}^\dagger(\boldsymbol{\Phi}_\txt{t},\boldsymbol{\Phi}_\txt{r}))\boldsymbol{\psi}\|_2^2}_{{J}(\boldsymbol{\Phi}_\txt{t},\boldsymbol{\Phi}_\txt{r})}, \label{p:a_unconstr}	
	\end{opteq}
	where we denote the objective function as $ {J}\in\mathbb{R}_+$, and define
	\begin{align*}
	\mathbf{K}(\boldsymbol{\Phi}_\txt{t},\boldsymbol{\Phi}_\txt{r}) &= \mathbf{A}(\mathbf{{F}}_\txt{t}(\boldsymbol{\Phi}_\txt{t})\diamond \mathbf{{F}}_\txt{r}(\boldsymbol{\Phi}_\txt{r})).
	\end{align*}
	Since \eqref{p:a_unconstr} is an unconstrained optimization problem with a continuous and differentiable objective function, we can find a local minimum of $ J $ using gradient descent. Straightforward computations show that the gradient is (see Appendix~\ref{ap:gd}):
	\begin{align}
	\nabla_{\boldsymbol{\Phi}_\txt{x}}{J}&=-2\Im\{\mathbf{{F}}_\txt{x}\circ\txt{mat}_{N_\txt{x}\times Q}((\partial_{\mathbf{K}}{J})(\partial_{\mathbf{{F}}_\txt{x}}\mathbf{K}))\},\label{eq:grad}
	\end{align}
	where the respective complex-valued matrix derivatives are
	\begin{align}
	\partial_{\mathbf{K}}{J}&= \txt{vec}^\txt{H}((\mathbf{K}\mathbf{K}^\dagger-\mathbf{I}_{V})\boldsymbol{\psi}(\mathbf{K}^\dagger\boldsymbol{\psi})^\txt{H})\label{eq:DDf}\\
	\partial_{\mathbf{{F}}_\txt{r}}\mathbf{K} &= (\mathbf{I}_{Q}\otimes\mathbf{A})((\mathbf{I}_{Q}\diamond \mathbf{{F}}_\txt{t})\otimes \mathbf{I}_{N_\txt{r}})\label{eq:DBrD}\\
	\partial_{\mathbf{{F}}_\txt{t}}\mathbf{K} &= (\mathbf{I}_{Q}\otimes\mathbf{A})(\mathbf{I}_{N_\txt{t}Q}\diamond (\mathbf{{F}}_\txt{r}\otimes \mathbf{1}_{N_\txt{t}}^\txt{T})).\label{eq:DBtD}
	\end{align}
	Here we define $ \partial_\mathbf{Z}\mathbf{X}(\mathbf{Z},\mathbf{Z}^*)=\frac{\partial \txt{vec}(\mathbf{X})}{\partial \txt{vec}^\txt{T}(\mathbf{Z})} \in\mathbb{C}^{AB\times CD} $, where  $\mathbf{X}\in\mathbb{C}^{A\times B} $ and $ \mathbf{Z} \in\mathbb{C}^{C\times D}$ \cite{hjorungnes2007complex}. Furthermore, $ (\cdot)^* $ denotes complex conjugation, $ (\cdot)^\txt{H} $ conjugate transposition, and $ \txt{mat}_{A\times B} $ reshapes an $AB$-dimensional vector into an $A \times B $ matrix. Given a step size $ \mu\in \mathbb{R}_{++} $, we update the gradient as $ \boldsymbol{\Phi}_\txt{x}\gets \boldsymbol{\Phi}_\txt{x}-\mu \nabla_{\boldsymbol{\Phi}_\txt{x}}{J} $. The update step is repeated until a maximum number of iterations $ k_{\max} $ or tolerance $ \varepsilon_{\max} $ is reached (see Algorithm~\ref{alg:gd}).
	\begin{algorithm}[h]
		\caption{Gradient descent algorithm for \eqref{p:h_alt}} \label{alg:gd}
		\begin{algorithmic}[1]		
			\Procedure{GradDesc}{$\mathbf{A},\boldsymbol{\psi},\mathbf{F}_\txt{r},\mathbf{F}_\txt{t},\mu,k_{\max},\varepsilon_{\max}$}
			\State $ \{k,\varepsilon,\mathbf{K}\}\gets \{0,\infty,\mathbf{A}(\mathbf{F}_\txt{t}\diamond\mathbf{F}_\txt{r})\} $
			\While{$k < k_{\max} \vee \varepsilon > \varepsilon_{\max}$}	
				\State Update derivative $ \partial_{\mathbf{K}}{J} $ using \eqref{eq:DDf}
				\For{$ \txt{x}\in\{\txt{t},\txt{r}\} $}
				\State Update derivative $ \partial_{\mathbf{{F}}_\txt{x}}\mathbf{K} $ using \eqref{eq:DBrD} or \eqref{eq:DBtD}
				\State Update gradient $ \nabla_{\boldsymbol{\Phi}_\txt{x}}J $ using \eqref{eq:grad}
				\State $ \mathbf{{F}}_\txt{x} \gets \exp( j(\measuredangle\mathbf{F}_\txt{x}-\mu\nabla_{\boldsymbol{\Phi}_\txt{x}}{J}))$
				\EndFor
				\State $ \mathbf{K}\gets \mathbf{A}(\mathbf{{F}}_\txt{t}\diamond\mathbf{{F}}_\txt{r}) $
				\State $ \varepsilon \gets \|(\mathbf{I}_V-\mathbf{K}\mathbf{K}^\dagger)\boldsymbol{\psi}\|_2 $
				\State $ k\gets k+1 $
			\EndWhile
			\State $ \{\mathbf{c}_\txt{r},\mathbf{c}_\txt{t}\} \gets \{\mathbf{K}^\dagger\boldsymbol{\psi},\mathbf{1}_{Q}\} $
			\State \Return $\mathbf{{F}}_\txt{r},\mathbf{{F}}_\txt{t},\mathbf{c}_\txt{r},\mathbf{c}_\txt{t}$
			\EndProcedure
		\end{algorithmic}
	\end{algorithm}
	The solution given by Algorithm~\ref{alg:gd} depends on the initialization of $ \mathbf{F}_\txt{x} $, since \eqref{p:a_unconstr} is a non-convex problem. Multiple random initializations may therefore be useful. Nevertheless, Algorithm~\ref{alg:gd} is guaranteed to improve upon an initial solution, provided that it is not a local minimum, and that the step size $ \mu $ is not too large. In the next section, we will see that good results can be achieved by using only a single random initialization of $ \mathbf{F}_\txt{x} $, and finding an appropriate step size $ \mu $ simply by trial-and-error.
		
	\section{Numerical examples} \label{sec:numerical}
	Next, we compare the PSFs of two analog beamformers with linear array geometries (Fig.~\ref{fig:architectures}). In particular, we consider a \emph{uniform linear array} (ULA) with $ N=11 $ elements, and a \emph{minimum-redundancy array} (MRA) \cite{moffet1968minimumredundancy,hoctor1996arrayredundancy} with $N = 7$ elements \cite{kohonen2014meet}. Both arrays span an aperture of $10\lambda/2 $, where $ \lambda/2 $ is the smallest inter-element spacing. Assuming ideal, identical, and omnidirectional transceiving elements, the transmit and receive steering vector becomes $ \mathbf{a}(\varphi)= \exp(j\pi \mathbf{d}\sin\varphi )$, where $ \mathbf{d}$ denotes the normalized array element positions. For the ULA: $\mathbf{d} = [-5, -4,\ldots, 5]^\txt{T}  $, and for the MRA: $\mathbf{d} = [-5, -4,-2, 0, 2, 4, 5]^\txt{T}  $. A Dolph-Chebyshev \cite{dolph1946acurrent} beampattern with $ -40 $ dB sidelobes is selected as the desired PSF. We initialize Algorithm~\ref{alg:gd} using uniformly distributed random phases, i.e., $ \mathbf{F}_\txt{x} = \exp(j\boldsymbol{\Phi}) $, where $ \Phi_{nq} \sim \txt{Uni}(0,2\pi)$. Furthermore, we set the step size to $ \mu = 10^{-3} $, the maximum number of iterations to $ k_{\max}=10^4 $, and the tolerance to $ \varepsilon_{\max} = 10^{-4} \|\boldsymbol{\psi}\|_2$. We sample the measurement matrix $ \mathbf{A} $ uniformly at $ V=99 $ azimuth angles between $-\pi/2$ and $ \pi/2 $. After finding the beamforming weights, we evaluate the realized PSF at $ 200 $ different angles in the same interval. 

	Fig.~\ref{fig:bp_lin_a}~(a) shows the desired and realized PSF of the ULA. Algorithm~\ref{alg:gd} yields a good approximation of the desired PSF using only a single component image. This approximation gets progressively better as $ Q $ is increased. Fig.~\ref{fig:bp_lin_a}~(b) shows that the MRA requires at least $ Q=2 $ component images to achieve an acceptable PSF. We note that the fully-digital ULA achieves the desired PSF using one component image, whereas the fully-digital MRA requires two components.\footnote{Using an alternating minimization algorithm \cite{rajamaki2019hybrid} with tolerance $ \varepsilon_{\max} $.} By Theorem~\ref{thm:M1F1Jinf}, the fully-analog beamformers then exactly achieve the fully-digital PSFs using $ Q= 4 $ (ULA), respectively $ Q= 8 $ (MRA) component images. Although the analog/digital MRA and ULA can achieve the same PSF, the MRA incurs a loss in array gain due to having fewer elements. 	
		\begin{figure}[]
		\centering
		\includegraphics[width=0.86\linewidth]{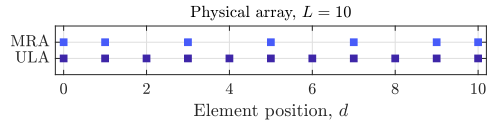}
		\caption{Uniform and sparse array configurations.}\label{fig:architectures}
	\end{figure}
	\begin{figure}[]
		\begin{minipage}[b]{.49\linewidth}
			\centering
			\centerline{\includegraphics[width=1\textwidth]{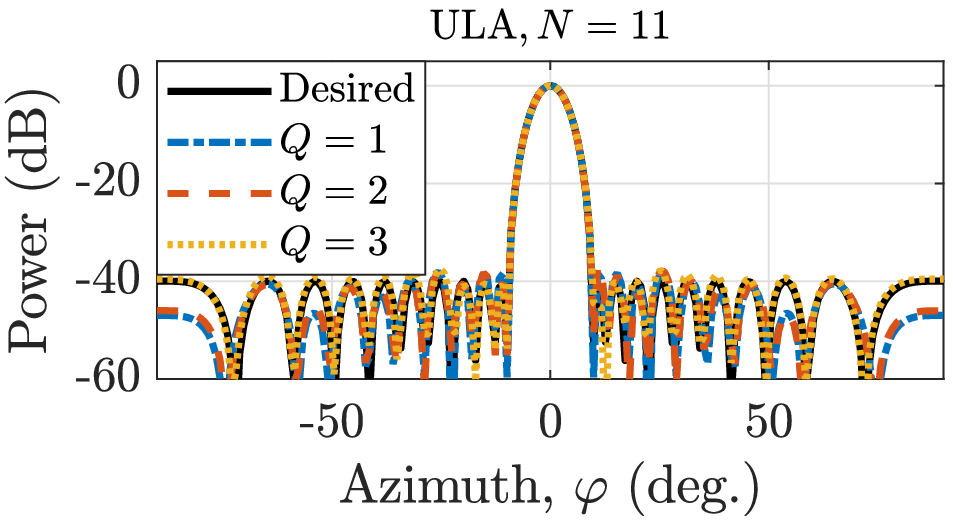}}
				\centerline{(a) ULA}\medskip
		\end{minipage}
		\hfill
		\begin{minipage}[b]{.49\linewidth}
			\centering
			\centerline{\includegraphics[width=1\textwidth]{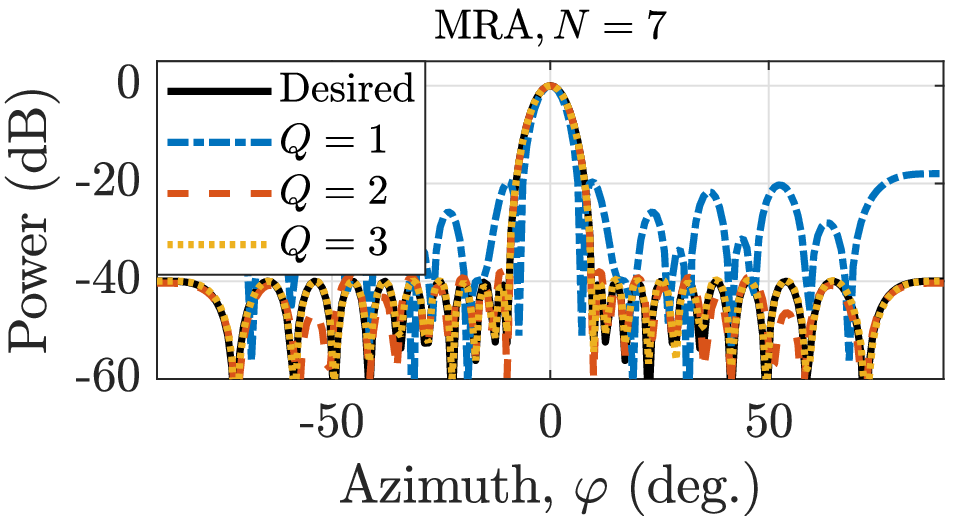}}
			\centerline{(b) MRA}\medskip
		\end{minipage}
		\caption{PSFs of analog beamformers. The (a) uniform array approximately achieves the desired PSF using one component image, whereas the (b) sparse array requires two.} \label{fig:bp_lin_a}
	\end{figure}

	\section{Conclusions} \label{sec:conclusions}
	This paper considered active imaging using phased arrays with an analog beamforming architecture consisting of continuous phase shifters. We proposed a gradient descent algorithm for jointly finding the transmit and receive element weights that achieve a desired PSF using as few transmissions as possible. Moreover, we derived a bound on the maximum number of transmissions required by such an array. We also demonstrated that combining analog beamforming with sparse arrays allows for significant reductions in the number of elements and RF-IF front-ends, without compromising main lobe width or side-lobe levels.
	
	\appendix
	\section{Derivation of gradient}\label{ap:gd}
	In this appendix, we derive the gradient in \eqref{eq:grad}. We are interested in $ \partial_{\boldsymbol{\Phi}_\txt{x}}{J} \in\mathbb{R}^{1\times N_\txt{x} Q}$, which by the chain rule is \cite{hjorungnes2007complex}
	\begin{align}
	\partial_{\boldsymbol{\Phi}_\txt{x}}{J} &= (\partial_{\mathbf{K}}{J})(\partial_{\boldsymbol{\Phi}_\txt{x}}\mathbf{K})+ (\partial_{\mathbf{K}^*}{J})(\partial_{\boldsymbol{\Phi}_\txt{x}}\mathbf{K}^*). \label{eq:DPf}
	\end{align}
	 Applying the chain rule again to \eqref{eq:DPf} yields
	\begin{align}
	\partial_{\boldsymbol{\Phi}_\txt{x}}\mathbf{K} &= (\partial_{\mathbf{{F}}_\txt{x}}\mathbf{K})( \partial_{\boldsymbol{\Phi}_\txt{x}}\mathbf{{F}}_\txt{x})+ (\partial_{\mathbf{{F}}_\txt{x}^*}\mathbf{K})( \partial_{\boldsymbol{\Phi}_\txt{x}}\mathbf{{F}}_\txt{x}^*)\label{eq:DPD}\\
	\partial_{\boldsymbol{\Phi}_\txt{x}}\mathbf{K}^* &= (\partial_{\mathbf{{F}}_\txt{x}}\mathbf{K}^*)( \partial_{\boldsymbol{\Phi}_\txt{x}}\mathbf{{F}}_\txt{x})+ (\partial_{\mathbf{{F}}_\txt{x}^*}\mathbf{K}^*)( \partial_{\boldsymbol{\Phi}_\txt{x}}\mathbf{{F}}_\txt{x}^*).\label{eq:DPDs}
	\end{align}
	Noting that $ \mathbf{K} $ only depends on $ \mathbf{F}_\txt{x} $, we have $ \partial_{\mathbf{{F}}_\txt{x}^*}\mathbf{K} = \partial_{\mathbf{{F}}_\txt{x}}\mathbf{K}^*= \mathbf{0} $. Substituting \eqref{eq:DPD} and \eqref{eq:DPDs} into \eqref{eq:DPf} then yields
	\begin{align*}
	\partial_{\boldsymbol{\Phi}_\txt{x}}{J}\!=\!(\partial_{\mathbf{K}}{J})(\partial_{\mathbf{{F}}_\txt{x}}\mathbf{K})( \partial_{\boldsymbol{\Phi}_\txt{x}}\mathbf{{F}}_\txt{x})\!+\!(\partial_{\mathbf{K}^*}{J})(\partial_{\mathbf{{F}}_\txt{x}^*}\mathbf{K}^*)( \partial_{\boldsymbol{\Phi}_\txt{x}}\mathbf{{F}}_\txt{x}^*).
	\end{align*}
	Combining $ \partial_{\mathbf{K}^*}{J} $ derived in \cite[Lemma~1]{jin2018hybridprecoding}, with the fact that $ \partial_{\mathbf{K}}{J}=(\partial_{\mathbf{K}^*}{J})^* $ yields \eqref{eq:DDf}. Equations \eqref{eq:DBrD} and \eqref{eq:DBtD} follow from identities $ \txt{vec}(\mathbf{K}) = (\mathbf{I}_{Q}\otimes\mathbf{A})\txt{vec}(\mathbf{{F}}_\txt{t}\diamond \mathbf{{F}}_\txt{r}) $ and	$ \txt{vec}(\mathbf{F}_\txt{t}\diamond \mathbf{F}_\txt{r})= ((\mathbf{I}_{Q}\diamond \mathbf{F}_\txt{t})\otimes \mathbf{I}_{N_\txt{r}})\txt{vec}(\mathbf{F}_\txt{r})
	= (\mathbf{I}_{N_\txt{t}Q}\diamond (\mathbf{F}_\txt{r}\otimes \mathbf{1}_{N_\txt{t}}^\txt{T})\txt{vec}(\mathbf{F}_\txt{t})$ \cite[Proposition~3.1.2]{roemer2013advanced}. Furthermore, we find that
	\begin{align}
		\partial_{\boldsymbol{\Phi}_\txt{x}}\mathbf{{F}}_\txt{x} &=j\txt{diag}(\txt{vec}(\mathbf{{F}}_\txt{x})).\label{eq:ap_DPB}
	\end{align}
	Confirming that $\partial_{\boldsymbol{\Phi}_\txt{x}}\mathbf{{F}}_\txt{x}^*= (\partial_{\boldsymbol{\Phi}_\txt{x}}\mathbf{{F}}_\txt{x})^* $ and $ \partial_{\mathbf{{F}}_\txt{x}^*}\mathbf{K}^*=(\partial_{\mathbf{{F}}_\txt{x}}\mathbf{K})^* $ allows us to simplify $ \partial_{\boldsymbol{\Phi}_\txt{x}}{J}$ using \eqref{eq:ap_DPB} as
	\begin{align}
	\partial_{\boldsymbol{\Phi}_\txt{x}}{J}&= 2\Re\{(\partial_{\mathbf{K}}{J})(\partial_{\mathbf{{F}}_\txt{x}}\mathbf{K})(\partial_{\boldsymbol{\Phi}_\txt{x}}\mathbf{{F}}_\txt{x})\}\nonumber\\
	&= 2\Re\{j(\partial_{\mathbf{K}}{J})(\partial_{\mathbf{{F}}_\txt{x}}\mathbf{K})\txt{diag}(\txt{vec}(\mathbf{{F}}_\txt{x}))\}\nonumber\\
	&=-2\Im\{(\partial_{\mathbf{K}}{J})(\partial_{\mathbf{{F}}_\txt{x}}\mathbf{K})\circ \txt{vec}^\txt{T}(\mathbf{{F}}_\txt{x})\}. \label{eq:DDPhi_final}
	\end{align}
	Finally, we obtain \eqref{eq:grad} by reshaping \eqref{eq:DDPhi_final} into a $ N_\txt{x}\times Q $ matrix.
\unappendix

	\bibliographystyle{IEEEtran}
	\bibliography{IEEEabrv,references}
	
\end{document}